\documentclass[11pt]{llncs}
\usepackage[left=26mm,right=26mm,top=30mm,bottom=32mm]{geometry}

\usepackage[noend]{algpseudocode}
\usepackage{algorithm}
\usepackage{verbatim}
\usepackage{color}
\renewcommand{\tabcolsep}{0.3cm}

\newcommand{\blue}[1] {\textcolor{blue}{#1}}

\usepackage{lineno}
\usepackage{latexsym,amsmath,amssymb}
\usepackage{graphicx,epsfig}
\usepackage{url}
\usepackage{multirow}
\usepackage[nocompress]{cite}

\usepackage{datetime}

\def\polylog{\mathrm{polylog}}

\renewcommand{\epsilon}{\varepsilon}

\newcommand{\PNF}{\mathrm{PNF}}
\newcommand{\LPN}{{\mathcal L}_{\textrm{PN}}}
\newcommand{\lext}{\mathcal{L}_{\textrm{ext}}}
\newcommand{\pext}{\textit{ext}}

\newcommand{\pn}{\textit{pnw}}
\newcommand{\crit}{\textit{crit}}

\makeatletter
\newcommand{\pushright}[1]{\ifmeasuring@#1\else\omit\hfill$\displaystyle#1$\fi\ignorespaces}
\newcommand{\pushleft}[1]{\ifmeasuring@#1\else\omit$\displaystyle#1$\hfill\fi\ignorespaces}
\makeatother

\newcommand{\floor}[1]{\lfloor #1 \rfloor}

\begin{document}

\title{Normal, Abby Normal, Prefix Normal}

\author{P\'eter Burcsi\inst{1} \and Gabriele Fici\inst{2} \and Zsuzsanna Lipt\'ak\inst{3} \and Frank Ruskey\inst{4} \and Joe Sawada\inst{5}}

\institute{Dept.\ of Computer Algebra,
E\"otv\"os Lor\'and Univ., Budapest, Hungary, 
\email{bupe@compalg.inf.elte.hu}
\and Dip.\ di Matematica e Informatica, University of Palermo, Italy, 
\email{gabriele.fici@math.unipa.it} 
\and Dip.\ di Informatica, University of Verona, Italy, 
\email{zsuzsanna.liptak@univr.it}
\and Dept.\ of Computer Science, University of Victoria, Canada, 
\email{ruskey@cs.uvic.ca}
\and School of Computer Science, University of Guelph, Canada, 
\email{jsawada@uoguelph.ca}
}

\date{}

\maketitle

\begin{abstract}
A prefix normal word is a binary word with the property that no substring has more 1s than the prefix of the same length. This class of words is important in the context of binary jumbled pattern matching. In this paper we present results about the number $\pn(n)$ of prefix normal words of length $n$, showing that $\pn(n) =\Omega\left(2^{n - c\sqrt{n\ln n}}\right)$ for some $c$ and $\pn(n) = O \left(\frac{2^n (\ln n)^2}{n}\right)$. We introduce efficient algorithms for testing the prefix normal property and a ``mechanical algorithm'' for computing prefix normal forms. We also include games which can be played with prefix normal words. In these games Alice wishes to stay normal but Bob wants to drive her ``abnormal'' -- we discuss which parameter settings allow Alice to succeed.
\end{abstract}


\section{Introduction}

Consider the binary word $w=10100110110001110010$. Does it have a substring of length $11$ containing exactly $5$ ones? 
In Fig.~\ref{fig:esempio} the word $w$ is represented by the black line (go up and right for a $1$, down and right for a $0$), while the grid points within the area between the two lighter lines form the {\em Parikh set} of $w$: the set of vectors $(x,y)$ s.t.\ some substring of $w$ contains exactly $x$ ones and $y$ zeros. Since the point $(5,6)$ lies within the area bounded by the two lighter lines, we see that the answer to our question is `yes'. (Don't worry, more detailed explanation will follow soon.) 
Now, this paper is about the lighter lines, called \emph{prefix normal words.}

\vspace{-.2cm}
\begin{figure}[h]
\begin{center}
  \includegraphics[height=35mm]{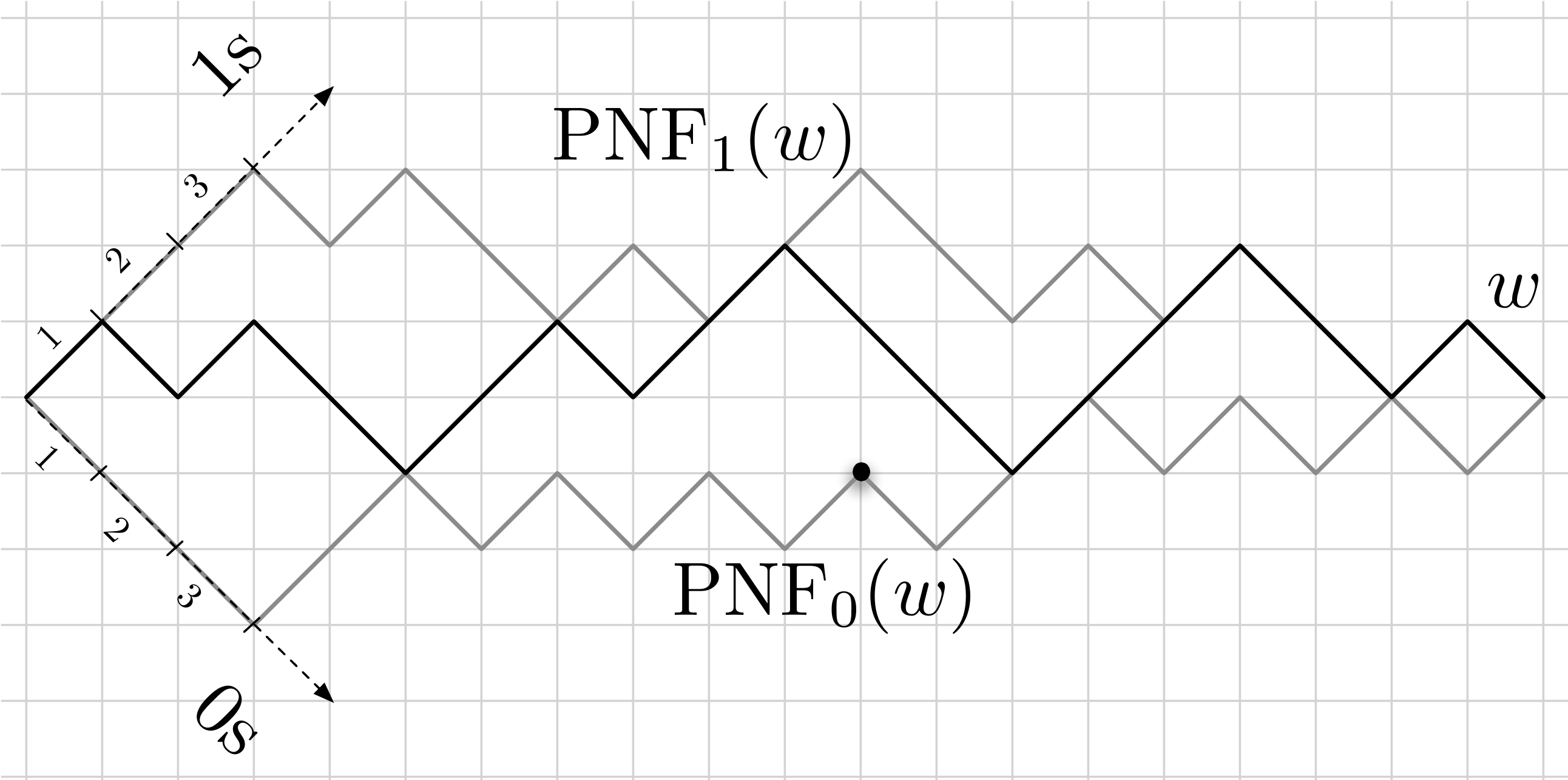}
  \caption{The word $w=10100110110001110010$ (dark line), its prefix normal forms $\PNF_{1}(w)=11101001011001010010$ and $\PNF_{0}(w)=00011010101011010101$ (lighter lines); the region between the two is the Parikh set of $w$; e.g.\  $w$ has a substring containing $5$ ones and $6$ zeros (black dot). Note that the axes are rotated by 45 degrees clockwise.}\label{fig:esempio}
\end{center}
\end{figure}

\noindent{\bf Prefix normal words%
:} A binary word $w$ is called {\em prefix normal} (with respect to $1$) if no substring of $w$ has more $1$s then the prefix of the same length\footnote{When not specified, we mean prefix normal w.r.t. 1.}. 
For example, $110101101100100$ is not prefix normal because it has a substring of length $5$ with $4$ ones, while the prefix of length $5$ has only $3$ ones. In~\cite{FL11} it was shown that to every word $w$, one can assign two prefix normal words, the {\em prefix normal forms} (PNF) of $w$ (w.r.t.\ $1$ and w.r.t.\ $0$), and that these are precisely the lines bounding $w$'s Parikh set from above (w.r.t. $1$) resp.\ from below (w.r.t.\ $0$), interpreted as binary words. 

\medskip

\noindent{\bf Prefix normal games:} Before we further elaborate on the connection between the initial problem and prefix normal words, let's see how well you have understood the definition. To this end, we define a two-player game. At the start of the game Alice and Bob have $n$ free positions. Alice moves first: she picks a position and sets it to $0$ or $1$. Then in alternating moves, they pick an empty position and set it. The game ends after $n$ moves. Alice wins if and only if the resulting binary word is prefix normal.

\begin{example} Here is an example run. We have $n=5$. Alice sets the first bit to $1$, then Bob sets the second bit to $0$. Now Alice sets the $4$th bit to $0$, and she has won, since whichever position Bob chooses, she will set the remaining position to $0$, thus ensuring that the word is prefix normal.\\

\begin{tabular}{ll @{\hspace{5mm}} *{5}{c} @{\hspace{1cm}} ll @{\hspace{5mm}} *{5}{c}}
1. & {\em start} & \_  & \_ & \_ & \_ & \_ & 3. & {\em Bob} & 1 & 0 & \_ & \_ & \_ \\
2. & {\em Alice} & 1 &  \_ & \_ & \_ & \_  & 4. & {\em Alice} & 1 & 0 & \_ & 0 & \_ 
\end{tabular}

\end{example}

The solution to the following exercise can be found in Section~\ref{sec:games}. 

\begin{exercise} 
Find the maximum $n$ such that Alice has a winning strategy.
\end{exercise}

\noindent{\bf Binary Jumbled Pattern Matching:} The problem of deciding whether a particular pair $(x,y)$ lies within the Parikh set of a word $w$ is known as {\em binary jumbled pattern matching}. There has been much interest recently in the {\em indexed version}, where an index for the Parikh set is created in a preprocessing step, which can then be used to answer queries fast. The Parikh set can be represented in linear space due to the {\em interval property} of binary strings: If $w$ has $k$-length substrings with $x_1$ resp.\ $x_2$ ones, where $x_1< x_2$, then it also has a $k$-length substring with $y$ ones, for every $x_1\leq y \leq x_2$ (folklore). Thus the Parikh set can be represented by storing, for every $1\leq k \leq |w|$,  the minimum and maximum number of $1$s in a substring of length $k$. Much recent research has focused on how to compute these numbers efficiently~\cite{CFL09,MR10,MR12,CLWY12,BFKL13,GG13,GHLW13}. The problem has also been extended to graphs and trees~\cite{GHLW13,CGGLLRT13}, to the streaming model~\cite{LLZ12}, and to approximate indexes~\cite{CLWY12}. There is also interest in the non-binary variant~\cite{BEL04,CFL09,KRR13}. A closely related problem is that of Parikh fingerprints~\cite{AALS03}.
Applications in computational biology include SNP discovery, alignment, gene clusters, pattern discovery, and mass spectrometry data interpretation~\cite{Boecker07,Benson03,BoeckerJMS08,DuhrkopLMB13,Parida06}. 

The current best construction algorithms for the linear size index for binary jumbled pattern matching run in $O(n^2/\log n)$ time~\cite{BCFL10,MR10}, for a word $w$ of length $n$, with some improvements for special cases (compressible strings~\cite{GHLW13,BFKL13}, bit-parallel operations~\cite{MR12,GG13})\footnote{Very recently, an algorithm with running time $n^2/2^{\Omega(\log n / \log\log n)^{1/2}}$ was presented~\cite{HLRW14}.}. As we will see later, computing the prefix normal forms is equivalent to creating an index for the Parikh set of $w$. Currently, we know no faster computation algorithms for the prefix normal forms than already exist for the linear-size index. However, should better algorithms be discovered, these would immediately carry over to the problem of indexed binary jumbled pattern matching.

\medskip

\noindent{\bf Testing:}  It turns out that even {\em testing} whether a given word is prefix normal is a nontrivial task. We can of course compute $w$'s prefix normal form, in $O(n^2/\polylog\ n)$ time using one of the above algorithms: obviously $w$ is prefix normal if and only if $w = \PNF(w)$. In~\cite{pn1}, we gave a generating algorithm for prefix normal words, which exhaustively lists all prefix normal words of a fixed length. The algorithm was based on the fact that prefix normal words are a bubble language, a recently introduced class of binary languages~\cite{RSW12,Wi09}. As a subroutine of our algorithm, we gave a linear time test for words which are obtained from a prefix normal word via a certain operation. In Section~\ref{sec:testing}, we present an algorithm to test whether an {\em arbitrary} word is prefix normal, based on similar ideas. 
Our algorithm is quadratic in the worst case but we believe it performs much better than other algorithms once some simple cases have been removed.  

We further demonstrate how using several simple linear time tests can be used as a filtering step, and conjecture, based on experimental evidence, that these lead to expected $O(n)$ time algorithms. 
But first the reader is kindly invited to try for herself.

\begin{exercise}
Decide whether the word $111010100110110011$ is prefix normal. 
\end{exercise}

\noindent{\bf Enumerating:} Another very interesting and challenging problem is the enumeration of prefix normal words. It turns out that even though the number of prefix normal words grows exponentially, the fraction of these words within all binary words goes to $0$ as $n$ goes to infinity. 
In Sections~\ref{sec:enum1} to~\ref{sec:enum3}, we present both asymptotic and exact results for prefix normal words, including generating functions for special classes and counting extensions for particular words. Some of the proofs in this part of the paper are rather technical: they will be available in the full version.

\medskip

\noindent{\bf Mechanical algorithm design:} We contribute to the area of mechanical algorithm design by presenting an algorithm for computing the Parikh set which uses the new \emph{sandbeach technique}, a technique we believe will be useful in many other applications (Sec.~\ref{sec:testing}).

\medskip
We would like to point out that prefix normal words, albeit similar in name, are not to be confused with so-called {\em Abby Normal} (a.k.a.\ {\em abnormal} or {\em AB normal}), words, or rather, brains, introduced in~\cite{MelBrooks}.---
And now it is time to wish you, the reader, as much fun in reading our paper as we had in writing it!

\section{Prefix normal words}

A {\em binary word} (or {\em string}) $w=w_1\cdots w_n$ over $\Sigma=\{0,1\}$ is a finite sequence of elements from $\Sigma$. Its length $n$ is denoted by $|w|$. For any $1\leq i\leq |w|$, the $i$-th symbol of a word $w$ is denoted by $w_{i}$. 
We denote by $\Sigma^n$ the words over $\Sigma$ of length $n$, and by $\Sigma^{*} = \cup_{n\geq 0} \Sigma^n$ the set of finite words  over $\Sigma$. The empty word is denoted by $\epsilon$. 
Let $w\in \Sigma^{*}$. If $w=uv$ for some $u,v\in\Sigma^{*}$, we say that $u$ is a \emph{prefix} of $w$ and $v$ is a \emph{suffix} of $w$. A \emph{substring} of $w$ is a prefix of a suffix of $w$. A {\em binary language} is any subset $\cal L$ of $\Sigma^*$. We denote by $|w|_c$ the number of occurrences in $w$ of character $c\in\{0,1\}$; $|w|_1$ is called the {\em density} of $w$.

Let $w\in \Sigma^*$. For $i=0,\ldots,n$, we set  $P(w,i) = |w_1\cdots w_i|_1$, the number of $1$s in the $i$-length prefix of $w$, and 
$F(w,i) = \max \{|u|_1 : u \text{ is a substring of } w \text{ and } |u|=i\}$, the maximum number of $1$s over all substrings of length $i$. 

Prefix normal words, prefix normal equivalence and prefix normal form were introduced in~\cite{FL11}. A word $w \in \{0,1\}^*$ is {\em prefix normal} (w.r.t.\ $1$) if, for all $1\leq i \leq |w|$, $F(w,i) = P(w,i)$. In other words, a word is prefix normal if no substring contains more $1$s than the prefix of the same length.

\begin{example}
We give all $23$ prefix normal words of length $n=6$:\\
000000,  
100000,  
100001,  
100010,  
100100,  
101000,  
101001, 
101010,  
110000,  
110001,
110010,
110011,
110100,  
110101, 
110110,  
111000,  
111001,  
111010,  
111011,  
111100,  
111101, 
111110,  
111111.
\end{example}
\vspace{-.2cm}

Two words $w,w'$ are {\em prefix normal equivalent} (w.r.t.\ $1$) if and only if $F(w,i) = F(w',i)$ for all $i$. Given $w\in\Sigma^*$, the {\em prefix normal form} (w.r.t.\ $1$) of $w$, $\PNF(w) =\PNF_1(w)$, is the unique prefix normal word $w'$ which is prefix normal equivalent (w.r.t.\ $1$) to $w$. 
Prefix normality w.r.t.\ $0$, prefix normal equivalence w.r.t.\ $0$, and $\PNF_0(w)$ are defined analogously.  When not stated explicitly, we are referring to the functions w.r.t.\ $1$. For example, the words $0000111$ and $1110000$ are prefix normal equivalent both w.r.t. $0$ and $1$. See \cite{FL11,pn1} for more examples.

In Fig.~\ref{fig:esempio}, we see an example string $w$ and its prefix normal forms. The interval property (see Introduction) can be graphically interpreted as vertical lines. The vertical line through point $(5,6)$ represents length-$11$ substrings: the grid points within the enclosed area are $(7,4), (6,5),$ and $(5,6)$, so all length-$11$ substrings have between $7$ and $5$ ones. We can interpret, for each length $k$, the intersection of the $k$th vertical line with the top grey line as the maximum number of $1$s, and with the bottom grey line as the minimum number of $1$s. Now it is easy to see that, passing from $k$ to $k+1$, this maximum, $F_1(w,\cdot)$, can either remain the same or increase by one. This means that the top grey line allows an interpretation as a binary word. A similar interpretation applies to the bottom line and prefix normal words w.r.t 0.

It should now be clear, also graphically, that the maximum number of $1$s for a substring of length $k$, $F(w,k)$, is precisely the number of $1$s in the $k$-length prefix of $\PNF_1(w)$ 
(the upper grey line); and similarly for the maximal number of $0$s (equivalently, the minimal number of $1$s) and $\PNF_0(w,k)$ (the lower grey line).
Moreover, these values can be obtained in constant time with constant-time rank-operations~\cite{Munro96,GHLW13}.

We list a few properties of prefix normal words that will be useful later. 

\begin{lemma}[Properties of prefix normal words~\cite{FL11}]\label{lemma:0suffix}
\begin{enumerate}
\item Every prefix of a prefix normal word is also prefix normal.
\item If $w$ is prefix normal, then $w0$ is also prefix normal.
\item Given $w$ of length $n$, it can be decided in $O(n^2)$ time whether $w$ is prefix normal.
\end{enumerate}
\end{lemma}

We denote the language of prefix normal words  by $\LPN$,  the number of prefix normal words of length $n$ by $\pn(n)$, and the number of prefix normal words  of length $n$ and density $d$ by $\pn(n,d)$. The first few values of the sequence $\pn(n)$ are  listed in \cite{sloane2}.

\section{Asymptotic bounds on the number of prefix normal words}\label{sec:enum1}

We give lower and upper bounds on the number of prefix normal words of length $n$.
Our lower bound on $\pn(n)$ is proved in Section \ref{sec:games}.

\begin{theorem}

\label{thm:lowBound}
There exists $c>0$ such that
\begin{equation}
\pn(n) = \Omega\left(2^{n - c\sqrt{n\ln n}}\right) = \Omega\left((2-\epsilon)^n\right) \qquad\textrm{for all } \epsilon >0.
\end{equation}
\end{theorem}

If we consider the length of the first 1-run, we obtain an upper bound.
\begin{theorem}
\label{thm:upBound}
For $n\geq 1$, we have $\pn(n) = O\left(\frac{2^n (\ln n)^2}{n}\right)=o(2^n)$.
\end{theorem}

\begin{proof}
Let $k=k(n)>0$ be a number to be specified later. Partition $\LPN\cap \Sigma^n \setminus\{0^n\}$ into two classes according to the length of the first 1-run.
\\\emph{Case 1:} If $w$ is prefix normal and the first 1-run's length is less than $k$, then there are no $k$ consecutive $1$s in $w$. Write $w$ as the concatenation of $\lfloor n/k \rfloor$ blocks of length $k$ and a final, possibly shorter block: $w=(w_1\ldots w_k)(w_{k+1}w_{k+2}\ldots w_{2k})\ldots$ For each block we have at most $2^k-1$ possibilities, so there can be at most $(2^k-1)^{\lceil n/k \rceil}$ words in this class.
\emph{Case 2:} The length of the first $1$-run in $w$ is at least $k$. Since the first $k$ symbols of $w$ are already fixed as $1$s, there can only be $2^{n-k} = 2^n/2^k$ words in this class.

If we balance the two cases by letting $k$ be the largest integer such that $2^k\cdot k^2 \cdot\ln 2\leq n$, then we have $k=\Theta(\ln n)$ and
\begin{equation*}
\pn(n)/2^n \leq \left(1-\frac{1}{2^k}\right)^{\lceil n/k \rceil} + \frac{1}{2^k} 
= \Theta\left(\frac{k^2}{n}\right) = \Theta\left(\frac{(\ln n)^2}{n}\right) = o(1),
\end{equation*}
as stated. \hfill \qed

\end{proof}

\section{Exact formulas for special classes of prefix normal words}\label{sec:enum2}

\subsubsection{Words with fixed density.}

We formulate an equivalent definition of the prefix normal property that will be useful in the enumeration of prefix normal words. Let $w=1w_2w_3\ldots w_n$ be a prefix normal word of density $d>0$. Denote by $r_1, r_2, \ldots, r_{d-1}$ the distances between consecutive occurrences of $1$ in $w$, and set $r_d$ so that  $\sum r_j = n$ holds. We can thus write $w=10^{r_1-1}10^{r_2-1}\ldots 10^{r_d-1}$. For $w=110100010$, we have $d=4$, $r_1=1$, $r_2=2$, $r_3=4$ and $r_4=2$. 
The prefix normal property is equivalent to requiring that for all $k$, one of the shortest substrings containing exactly $k$ ones is a prefix. This gives us the following lemma.

\begin{lemma}
\label{lemma:equivDefPn}

The binary word $w$ is prefix normal if and only if the following inequalities hold:
\begin{equation*}
\begin{array}{rcll}
r_1 &\leq& r_j \qquad & j=2, 3, \ldots, d-3, d-2, d-1 \\
r_1+r_2 &\leq& r_j+r_{j+1} \qquad & j=2, 3, \ldots , d-3, d-2 \\
r_1+r_2+r_3 &\leq& r_j+r_{j+1}+r_{j+2} \qquad &j=2, 3, \ldots , d-3 \\
&\vdots&& \vdots\\
r_1+r_2+\cdots + r_{d-2} &\leq& r_j+r_{j+1}+\cdots + r_{d-1} \qquad & j=2
\end{array}
\end{equation*}
\end{lemma}

\begin{lemma}
\label{lemma:genFunc}
For $d=0,\ldots,6$, we have the generating functions $f_d(x)=\sum_{n=1}^{\infty} \pn(n,d)x^n$:
\begin{center}
\begin{minipage}{7.2cm}
\begin{eqnarray*}
f_0(x) &=&  \frac{1}{1-x}\\
f_1(x) &=&  \frac{x}{1-x}\\
f_2(x) &=&  \frac{x^2}{(1-x)^2}\\
f_3(x) &=& \frac{x^3}{(1-x^2)(1-x)^2}
\end{eqnarray*}
\end{minipage}
\begin{minipage}{7.2cm}
\begin{eqnarray*}
f_4(x) &=& \frac{x^4}{(1-x^3)(1-x)^3}\\[2mm]
f_5(x) &=&\frac{x^5(1+x+x^2)}{(1-x^4)(1-x^2)^2(1-x)^2}\\[2mm]
f_6(x) &=& \frac{x^6(1+x+x^2+x^3)}{(1-x^5)(1-x^3)(1-x^2)(1-x)^3}
\end{eqnarray*}
\end{minipage}
\end{center}
\end{lemma}

Similar formulas can be derived for $\pn(n, n-d)$ for small values of $d$.
Unfortunately, no clear pattern is visible for $f_d(x)$ that we could use for calculating $\pn(n)$.

\subsubsection{Words with a fixed prefix.}
We now fix a prefix $w$ and give enumeration results on prefix normal words with prefix $w$. Our first result indicates that we have to consider each $w$ separately.

\begin{definition}
If $w$ is a binary word, let $\lext(w) = \{ w' : ww' \textrm{ is prefix normal } \}$, and $\lext(w, m) = \lext(w)\cap \Sigma^{|w|+m}$.
Let $\pext(w, m, d) = |\{ w' : ww' \textrm{ is prefix normal of length } |w|+m \textrm{ and density }d \}|$, and $\pext(w, m) = |\lext(w, m)|$.
\end{definition}

\begin{lemma}
\label{lemma:extLang}
Let $v, w \in 1\{0,1\}^*$ be both prefix normal. If $v\neq w$ then $\lext(v) \neq \lext(w)$.
\end{lemma}

We were unable to prove that the growth of these two extension languages also differ.
\begin{conjecture}
Let $v, w \in 1\{0,1\}^*$ be both prefix normal. If $v\neq w$ then the infinite sequences $(\pext(v, m))_{m\geq1}$ and $(\pext(w,m))_{m\geq 1}$ are different.
\end{conjecture}

The values $\pext(w, m, d)$ seem hard to analyze. We give exact formulas for a few special cases of interest. Using Lemma \ref{lemma:equivDefPn}, it is possible to give formulas similar to those in Lemma \ref{lemma:genFunc} for $\pext(w, m, d)$ for fixed $w$ and $d$. We only mention one such result. 

\begin{lemma} 
For $1\leq d\leq n$ we have $\pext(10, n+d-3, d)   = \pn(n, d)$.
\end{lemma}

\begin{proof}
Let $w$ be an arbitrary prefix normal word of length $n$ and density $d$ with $1$ as its first symbol. Insert a $0$ before each subsequent occurrence of $1$. It is easy to see that this operation creates a bijection between the two sets that we want to enumerate.
\hfill \qed
\end{proof}

The following lemma lists exact values for $\pext(w, |w|)$ for some infinite families of words $w$.

\begin{lemma}
Let $F(n)$ denote the $n$th Fibonacci number: $F(1)=F(2)=1$ and $F(n+2) = F(n+1)+F(n)$. Then for all values of $n$ where the exponents are nonnegative, we have the following formulas:
\begin{center}
\begin{minipage}{7.2cm}
\begin{eqnarray*}
&&\pext(0^n,n) = 1\\[1mm]
&& \pext(1^n, n) = 2^n \\[1mm]
&& \pext(1^{n-1}0, n) = 2^n - 1 \\[1mm]
&& \pext(1^{n-2}01,n) = 2^n - 5 \\[1mm]
&& \pext(1^{n-2}00,n) = 2^n - (n+1)
\end{eqnarray*}
\end{minipage}
\begin{minipage}{7.2cm}
\begin{eqnarray*}
&&\pext((10)^{\frac n2}, n) = F(n+2) \textrm { if }n \textrm{ is even }\\[2mm]
&&\pext((10)^{\lfloor \frac n2\rfloor}1,n) = F(n+1) \textrm{ if }n\textrm{ is odd}\\[2mm]
&&\pext(10^{n-2}1,n) = 3\\[2mm]
&&\pext(10^{n-1},n) = n+1
\end{eqnarray*}
\end{minipage}
\end{center}

\end{lemma}

\begin{proof}
For $w=1^n$, $w=1^{n-1}0$, $w=1^{n-2}01$ and $w=1^{n-2}00$, it is easy to count those extensions that fail to give prefix normal words. Similarly, for $w=10^{n-2}1$, $w=10^{n-1}$ and $w=0^n$, counting the extensions that give prefix normal words gives the results in a straightforward way.

Let $n$ be even. For $w=(10)^{\frac n2}$, note that $ww'$ is prefix normal if and only if $w'$ avoids $11$. The number of such words is known to equal $F(n+2)$. For $n$ odd, the argument is similar.
\hfill \qed
\end{proof}

\section{Experimental results about prefix normal words}\label{sec:enum3}

We consider extensions of prefix normal words by a single symbol to the right. It turns out that this question has implications for the enumeration of prefix normal words. 

\begin{definition}
We call a prefix normal word $w$ extension-critical if $w1$ is not prefix normal. Let $\crit(n)$ denote the number of extension-critical words in $\LPN\cap \Sigma^n$.
\end{definition}

\begin{lemma} For $n\leq 1$ we have
\begin{equation}
\pn(n) = 2\pn(n-1) - \crit(n-1) = \pn(n-1)\left(2-\frac{\crit(n-1)}{\pn(n-1)}\right).
\end{equation}
From this it follows that
\begin{equation}
\label{eq:critProduct}
\pn(n) = 2\prod_{i=1}^{n-1} \left(2-\frac{\crit(i)}{\pn(i)}\right).
\end{equation}
\end{lemma}

From Theorem \ref{thm:lowBound} we have:
\begin{lemma} For $n$ going to infinity, $\lim\inf \crit(n)/\pn(n)  =  0$.
\end{lemma}

We conjecture that in fact the ratio of extension-critical words converges to $0$. We study the behavior of $\crit(n)/p(n)$ for $n\leq 49$. The left plot in Fig.~\ref{figCrit} shows the ratio of extension-critical words for $n\leq 49$. These data support the conjecture that the ratio tends to $0$. Interestingly, the values decrease monotonically for both odd and even values, but we have $\crit(n+1)/\pn(n+1) > \crit(n)/\pn(n)$ for even $n$. We were unable to find an explanation for this.

The right plot in Fig.~\ref{figCrit} shows the ratio of extension-critical words multiplied by $n/\ln n$. Apart from a few initial data points, the values for even $n$ increase monotonically and the values for odd $n$ decrease monotonically, and the values for odd $n$ stay above those for even $n$. 

\begin{conjecture} 
Based on empirical evidence, we conjecture the following:
\begin{eqnarray}
\crit(n) &=& \pn(n) \Theta(\ln n / n) ,\\
\pn(n) &=& 2^{n-\Theta((\ln n)^2)} .
\end{eqnarray}
\end{conjecture}

Note that the second estimate follows from the first one by  \eqref{eq:critProduct}.

\begin{figure}
\begin{minipage}[c]{8cm}
\begin{center}
\includegraphics[scale=0.3]{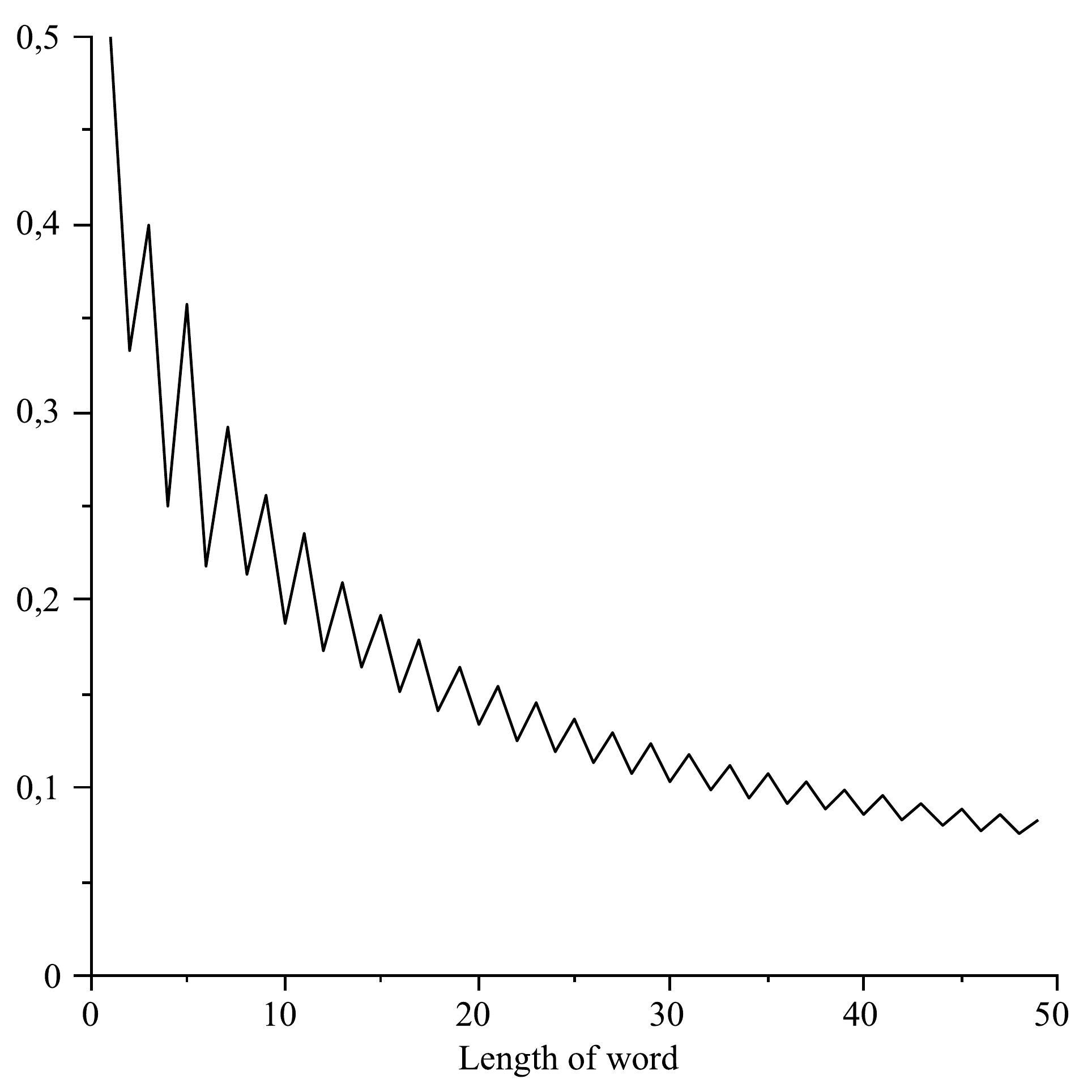}
\end{center}
\end{minipage}
\begin{minipage}[c]{8cm}
\begin{center}
\includegraphics[scale=0.3]{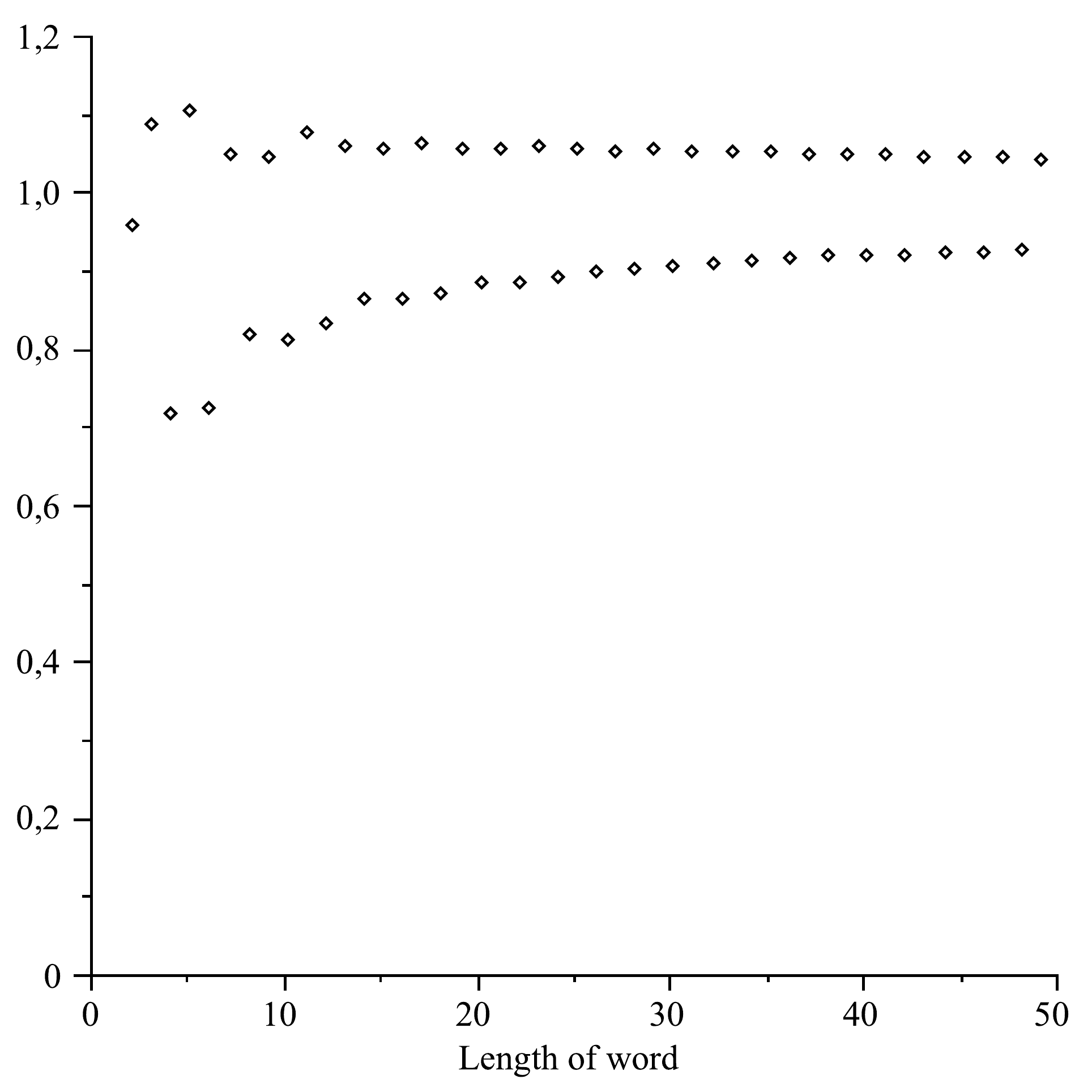}
\end{center}
\end{minipage}
\caption{The ratio $\frac{\crit(n)}{\pn(n)}$ (left), and the value $\frac{\crit(n)}{\pn(n)}\cdot\frac{n}{\ln n}$ (right).} 
\label{figCrit}
\end{figure}

\section{Prefix Normal Games}
\label{sec:games}
\textbf{Variant 1: Prefix normal game starting from empty positions.} See Introduction.

\begin{lemma}
\label{lemma:game1}
For $n\geq 7$ Bob has a winning strategy in the game starting from empty positions.
\end{lemma}

\noindent\textbf{Variant 2: Prefix normal game with blocks.} The game is played as follows. Now a block length of $2k$ is also specified, and we require that $2k$ divides $n$. The first $4k$ symbols are set to $1$ before the game starts (in order to give Alice a fair chance). Divide the remaining empty positions into blocks of length $2k$. Then Bob starts by picking a block with empty positions, and setting half of the positions of the block arbitrarily. Alice moves next and she sets the remaining $k$ positions in the same block as she wants. Now this block is completely filled. Then Bob picks another block, fills in half of it, etc. Iterate this process until every position is filled in.

\begin{lemma}
Alice has a winning strategy in the game with blocks, for any $k\geq 1$.
\end{lemma}

\begin{proof}
Alice can always achieve that the current block contains exactly $k$ $1s$ and $k$ $0$s. Now consider a substring $v$ of length $m$ of the word $w=1^{4k}u$ that is obtained in the end. We have to show that the prefix of the same length has at least as many $1$s. Clearly, only $m\geq 4k$ has to be considered, and we can also assume that $v$ starts after position $4k$. The substring $v$ contains some $2k$-blocks in full, and some others partially. Let $p := \floor{\frac{m}{2k}}$, then $|v|_1 \leq (p+1)k \leq \frac{m}{2} + k$, while the number of $1$s in the prefix of length $m$ is at least $4k + (p-2)k \geq \frac m2 +k$, as claimed. 
\hfill \qed
\end{proof}
As a corollary, we can prove the lower bound in Theorem \ref{thm:lowBound}.

\begin{proof} (of Theorem~\ref{thm:lowBound}). There are at least as many prefix normal words of length $n$ as there are distinct words resulting after a game with blocks that Alice has won using the above strategy. Note that with this strategy, each block has exactly $k$ many $0$s and Bob is free to choose their positions within the block. Moreover, for different choices of $0$-positions by Bob, the resulting words will be different. So overall, Bob can achieve at least $\binom{2k}{k}^{(n-4k)/2k}$ different outcomes. If we set $k=\floor{\sqrt{n \log n}}$, and note that for $2k$ not dividing $n$, we can use $\pn(n)\geq \pn(\floor{n/2k}\cdot 2k)$, then we obtain: 
$-\ln (\pn(n)/2^n) = O(\sqrt{n\ln n}),$
and the statement follows.
\hfill \qed
\end{proof}

\section{Construction and testing algorithms}\label{sec:testing}

In this section, for strings $w\neq 1^n$, we use the notation $w=1^s0^t\gamma$, with $s\geq 0, t>0$ and $\gamma\in 1\Sigma^* \cup \{\epsilon\}$. Note that this notation is unique. We call $1^s0^t$ the {\em critical prefix} of $w$.

\subsection{A mechanical algorithm for computing the prefix normal forms}

We now present a {\em mechanical} algorithm for computing the prefix normal form of a word $w$. It uses a new algorithm technique we refer to as {\em sandy beach technique}, a technique that we think will be useful for many other similar problems. 

First observe that if you draw your word $w$ as in Fig.~\ref{fig:esempio}, then the Parikh set of $w$ will be the region spanned by drawing all the suffixes of $w$ starting from the origin. As we know, the prefix normal forms of $w$ will be the upper and the lower contour of the Parikh set, respectively. This leads to the following algorithm, that we can implement in any sand beach---for example, Lipari's Canneto (Fig.~\ref{fig:zollstock}).

Take a folding ruler (see Fig.~\ref{fig:zollstock}) and fold it in the form of your word. Now designate an origin in the sand. Put the folding ruler in the sand so that its beginning coincides with the origin. Next, move it backwards in the sand such that the position at the beginning of the $(n-1)$-length suffix coincides with the origin; then with the next shorter suffix and so on, until the right end of the folding ruler reaches the origin. The traced area to the right of the origin is the Parikh set of $w$, and its top and bottom boundaries, the prefix normal forms of $w$ (that you can save by taking a photo).

{\em Analysis: } The algorithm requires a quadratic amount of sand, but can outperform existing ones in running time if implemented by a very fast person.

\begin{figure}
\begin{center}
\begin{minipage}{6cm}
\includegraphics[width=\textwidth]{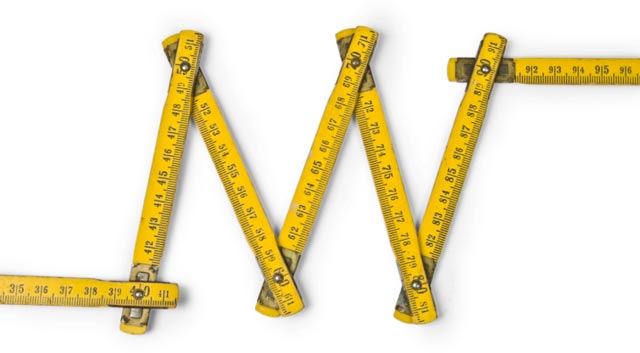}
\end{minipage}
\begin{minipage}{1cm}
\hspace{2mm}
\end{minipage}
\begin{minipage}{6cm}
\includegraphics[width=\textwidth]{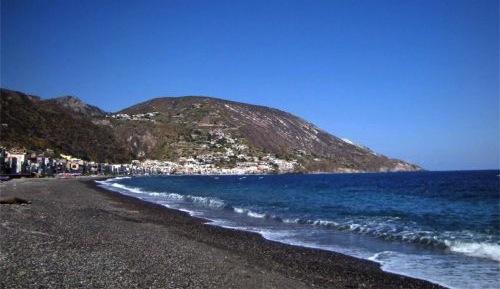}
\end{minipage}
\end{center}
\caption{The folding ruler used and a sandy beach (here the beautiful Liparis's Canneto black sand beach) in our mechanical prefix normal construction algorithm.\label{fig:zollstock}}
\end{figure}

\subsection{Testing algorithm}

It can be tested easily in $O(n^2)$ time if a word is prefix normal, by computing its $F$-function and comparing it to its prefixes; several other quadratic time tests were presented in~\cite{FL11}. Currently, the fastest algorithms for computing $F$ run in worst-case $O(n^{2}/\polylog\ n)$ time (references in the Introduction). Here we present another algorithm, which, although $O(n^2)$ in the worst-case, we believe could well outperform other algorithms when iterated on prefixes of increasing length. 

Given a word $w$ of length $n$ and density $d$, $w=1^s0^t\gamma$. Since the cases $d=0,n$ are trivial, we assume $0<d<n$.
Notice that, then, in order for $w$ to be prefix normal, $s>0$ must hold.
Now build a sequence of words $v_{0},v_{1},\ldots,v_{d-s}$, where $v_{0}=1^{d}0^{n-d}$ and $v_{d-s}=w$, in the following way: for every $i$, $v_{i+1}$ is obtained from $v_{i}$ by swapping the positions $d-i$ and $j$, where $j$ is the rightmost mismatch between $v_{i}$ and $w$. 
So for example, if $w=110100101$, we have the following sequence of words: $111110000$, $111100001$, $111000101$, $110100101$.

The following lemma follows straightforwardly from the results of~\cite{pn1}: 

\begin{lemma}
Given $w\in\Sigma^n$ with $|w|_1=d$, and the sequence $v_{0}=1^{d}0^{n-d}, v_{1}, \ldots, v_{d-s}=w$, we have that $w$ is prefix normal if and only if every $v_{i}$ is. 
\end{lemma}

Moreover, as was shown there, it can be checked efficiently whether these strings are prefix normal. We summarize in the following lemma, and give a proof sketch and an example. 

\begin{lemma}[from \cite{pn1}]\label{lemma:childpnf}
Given a prefix normal word $w=1^s0^t\gamma$. Let $w'=1^{s-1}0^i10^{t-i}\gamma$, then it can be decided in linear time whether $w'$ is prefix normal.
\end{lemma}

We will give an intuition via a picture, see Fig.~\ref{fig:proof_childpnf}. If $w'$ is not prefix normal, then there must be a $k$ and a substring $u$ of length $k$ s.t.\ $u$ has more $1$s than the prefix of length $k$. It can be shown that it suffices to check this for one value of $k$ only, namely for $k=s-1+t$, the  length of the critical prefix length of $w'$. The number of $1$s in this prefix is $s-1$. Now if such a $u$ exists, then it is either a substring of $\gamma$, in which case $F(\gamma,k)>s-1$; or it is a substring which contains the position of the newly swapped $1$ (both in grey in the third line). This latter case can be checked by computing the number of $1$s in the prefix of the appropriate length of $\gamma$ (in slightly darker grey) and checking whether it is greater than $s-2$.

\begin{figure}
\begin{center}
\includegraphics[width=\textwidth]{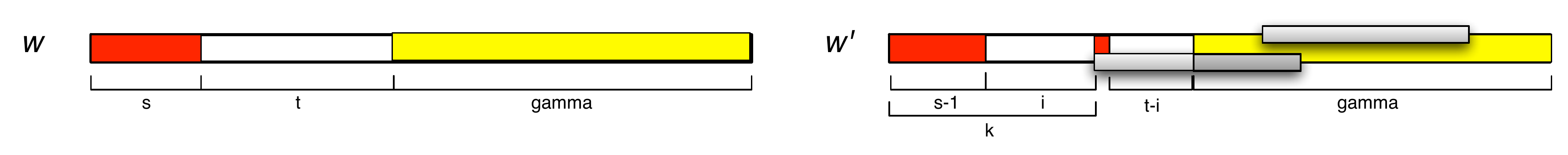}
\end{center}
\vspace{-1cm}
\caption{Proof of Lemma~\ref{lemma:childpnf}.\label{fig:proof_childpnf}}
\end{figure}

Thus, for $i = 1,\ldots, d-s$, we test if $v_{i+1}$ is prefix normal. If at some point, we receive a negative answer, then the test returns NO, otherwise it returns YES. Additional data structures for the algorithm are the $F$-function, which is updated to the current suffix following the critical prefix, up to the length of the next critical prefix (in linear time); and a variable $z$ containing the number of $1$s in the appropriate length prefix of $\gamma$. 

\medskip

{\em Example:} 
We test whether the word $w=110101101100100$ is prefix normal. 

\medskip

$
\begin{array}{l @{\hspace{5mm}} r @{\hspace{5mm}} r @{\hspace{5mm}} r @{\hspace{5mm}} r @{\hspace{5mm}} r @{\hspace{5mm}} l}
w  & 110101101100100 & \gamma & k & F(k) & z & F\\
\hline
v_1& 111111110000000 & \epsilon & 12 & 0 & 0 & 000000000000 \\
v_2& 111111100000100 & 100 & 9 & 1 & 1 & 111111111\\
v_3& 111111000100100 & 100100 & 8 & 2 & 2 & 11122222\\
v_4& 111110001100100 & 1100100 & 6 & 3 & 2 & 122233\\
v_5& 111100101100100 & 101100100 & 5 & 3 & 3 & 12233
\end{array}
$

\medskip

At this point we have $z + 1 = 4 > 3= s-1$  and therefore, we stop. Indeed, we can see that the next word to be generated, $v_6 = 1110001101100100$ is not be prefix normal, since it has a substring of length $5$ with $4$ ones, but the prefix of length $5$ has only $3$ ones. 

{\em Analysis: } The running time of the algorithm is $O(\sum_{i=d-s}^d p_i)$ in the worst case, where the $p_i$ are the positions of the $1$s in $w$, so in the worst case quadratic. 

\medskip

{\bf Iterating version.} The algorithm tests a condition on the suffixes starting at the $1$s, in increasing order of length, and compares them to a prefix where the remaining $1$s but one are in a block at the beginning. This implies that for some $w$ which are not prefix normal, e.g.\ $w=101^n, n>1$, the algorithm will stop very late, even though it is easy to see that the word is not prefix normal. This problem can be eliminated by running some linear time checks on the word first; the power of this approach will be demonstrated in the next section. 

Since we know that a word $w$ is prefix normal iff every prefix of $w$ is, we have that a word which is \emph{not} prefix normal has a shortest non-prefix-normal prefix. We therefore adapt the algorithm in order to test the prefix normality on the prefixes of $w$  of length powers of $2$, in increasing order. In the worst case, we apply the algorithm $\log n$ times. Since the test on the prefix of length $2^{i}$ takes $O(2^{2i})$ time, we have an overall $\sum_{i=0}^{\log n}O(2^{2i})=O(n^{2})$ worst case running time, so no worse than the original algorithm. 

We believe that our algorithm will perform well on strings which are ``close to prefix normal'' in the sense that they have long prefix normal prefixes, or they have passed the filters, i.e.\ that it will be expected strongly subquadratic, or even linear, time even on these strings.

\subsection{Membership testing with linear time filters}
\label{sec:combgenpnw}

In this section, we provide a two-phase membership tester for prefix normal words.  Experimental evidence indicates that on \emph{average} its running time is $O(n)$.

Suppose there is an $O(n)$ test that can be used to reject $2^n - 2^n/n$ of the binary strings outright (Phase I).   For the remaining $2^n/n$ strings,  apply the worst case $O(n^2)$ algorithm (Phase II).  This gives an $O(n)$-amortized time  algorithm when taken over all $2^n$ strings.  For such a two-phase  approach, let $M$ denote the strings not rejected by the first phase.  We are interested in the ratio 
${n M}/{2^n}.$ 
As $n$ grows, if it appears as though this ratio is bounded by a constant, then we would conjecture that such a membership tester runs in $O(n)$ average case time.

First we try a trivial $O(n)$ test:  a string will \emph{not} be prefix-normal if the longest substring of 1s is not at the prefix.  Applying this test as the first phase, the resulting ratios for some increasing values of $n$ are given in Table~\ref{tab:ratios}(a).  Since the ratios are increasing as $n$ increases, we require a more advanced rejection test.

\begin{table}[t]
\label{tab:ratios}
\begin{center}
\begin{tabular}{c|c|c|c|c|c|c|c|c}
$n$ & 10 & 12 & 14 & 16 & 18 & 20 & 22 & 24 \\ \hline
(a) & 2.500 & 2.561 & 2.602 & 2.631 & 2.656 & 2.675 & 2.693 & 2.708 \\ \hline
(b) & 2.168 & 2.142 & 2.121 & 1.106 & 2.093 & 2.083 & 2.075 & 2.067 \\
\end{tabular}

\end{center}
  \caption{(a) Ratios from the trivial rejection test.  \ (b) Ratios by adding secondary rejection test.}
\end{table}

The next attempt uses a more compact \emph{run-length} representation for $w$.  Let $w$ be  represented by a series of $c$ blocks, which are maximal substrings of the form $1^*0^*$. Each block $B_i$ is composed of two integers $(s_i, t_i)$ representing the number of 1s and 0s respectively. For example, the string  
11100101011100110
can be represented by $B_1B_2B_3B_4B_5 = (3, 2)(1, 1)(1,1)(3,2)(2,1)$.  Such a representation can easily be found in $O(n)$ time.  
A word $w$ will \emph{not} be prefix normal word if it contains a substring of the form $1^{i}0^j1^k$ such that
$i+j+k \leq s_1 + t_1$ and $i+k > s_1$ (the substring is no longer, yet has more 1s than the critical prefix). Thus, a word will not be prefix normal, if for some $2 \leq i \leq c$:
$$s_{i-1} + t_{i-1} + s_i \leq  s_1 + t_1     \text{  \  \ and  \  \ }  s_{i-1} + s_{i} > s_1.$$
By applying this additional test in our first phase, we obtain algorithm {\sc MemberPN}($w$), consisting of the two rejection tests, followed by any simple quadratic time algorithm. 

The ratios that result from this algorithm are given in Table~\ref{tab:ratios}(b).  Since the ratios are decreasing as $n$ increases, we make the following conjecture.

\begin{conjecture}
The membership tester {\sc MemberPN}($w$) for prefix normal words \blue{funs} in average case $O(n)$-time.
\end{conjecture}
We note that there are several other trivial \emph{rejection} tests that run in $O(n)$ time, however these two were sufficient to obtain our desired experimental results.

\textbf {Acknowledgements.} We thank Ferdinando Cicalese who pointed us to~\cite{MelBrooks} and thus contributed to the {\em fun} part of our paper.

\bibliographystyle{abbrv}

\end{document}